\documentclass{iopart}

\expandafter\let\csname equation*\endcsname\relax
\expandafter\let\csname endequation*\endcsname\relax

\usepackage{amsmath}
\usepackage{amssymb,amsfonts,amsthm}
\usepackage[T1]{fontenc}
\usepackage[utf8]{inputenc}
\usepackage[francais,english]{babel} 
\usepackage{enumitem}

\usepackage{esint}
\usepackage{dsfont}

\newtheorem{theorem}{Theorem}

\newtheorem{lemma}{Lemma}

\eqnobysec

\def\tr{{\rm tr \,}}
\def\Tr{{\rm Tr \,}}

\def\N{{\mathbb N}}

\def\R{{\mathbb R}}
\def\C{{\mathbb C}}

\def\bA{{\bold A}}

\def\bB{{\bold B}}
\def\bR{{\bold R}}
\def\bm{{\bold m}}
\def\br{{\bold r}}

\def\be{{\bold e}}

\def\bx{{\bold x}}
\def\by{{\bold y}}

\def\bsigma{{\bold \sigma}}

\def\bnull{{\bold 0}}

\def\fs{{\mathfrak{S}}}

\def\ri{{\mathrm{i}}}
\def\rd{{\, \mathrm{d}}}
\def\re{{\mathrm{e}}}

\def\Re{{\mathrm{Re}}}

\def\cB{{\mathcal B}}
\def\cC{{\mathcal C}}
\def\cD{{\mathcal D}}
\def\cE{{\mathcal E}}

\def\cJ{{\mathcal J}}
\def\cM{{\mathcal M}}
\def\cP{{\mathcal P}}

\def\cW{{\mathcal W}}

\def\b1{{\mathds 1}}
\def\bI{{\mathds I}}

\def\spinup{{\uparrow }}
\def\spindown{{\downarrow }}

\newcommand{\bra}{\langle}
\newcommand{\ket}{\rangle}

\newcommand{\xc}{\mathrm{xc}}
\newcommand{\loc}{\mathrm{loc}}
\newcommand{\LSDA}{\mathrm{LSDA}}
\newcommand{\LDA}{\mathrm{LDA}}
\renewcommand{\Im}{\mathrm{Im}}

\begin{document}

\title[Existence of minimizers for Kohn-Sham within the LSDA]{Existence of minimizers for Kohn-Sham within the Local Spin Density Approximation}

\author{David Gontier}

\address{Universit\'e Paris Est, CERMICS (ENPC), INRIA\\
F-77455 Marne-la-Vall\'ee,
France}

\ead{david.gontier@cermics.enpc.fr}



\begin{abstract}
The purpose of this article is to extend the work by Anantharaman and Canc\`es \cite{Anantharaman2009}, and prove the existence of minimizers for the spin-polarized Kohn-Sham model in the presence of a magnetic field within the local spin density approximation. We show that for any magnetic field that vanishes at infinity, the existence of minimizers is ensured for neutral or positively charged systems. The proof relies on classical concentration-compactness techniques.
\end{abstract}



\section{Introduction} 

The density functional theory (DFT) introduced in 1965 by Hohenberg and Kohn \cite{Hohenberg1964} is a very popular tool in modern quantum chemistry. This theory transforms the high-dimensional Schr\"odinger problem into a low-dimensional one, hence computationally solvable. The price to pay is the introduction of the so-called exchange-correlation (xc) energy term, which is unknown. Throughout the literature, several different approximations of this energy can be found. The first successful one, and still broadly used nowadays, was proposed by Kohn and Sham \cite{Kohn1965}, and is called the local density approximation (LDA). The mathematical properties resulting of the Kohn-Sham LDA are still not fully understood. Proving the existence of minimizers is made difficult by the non-convexity of the problem due to the LDA term. Using concentration compactness techniques introduced by Lions \cite{Lions1984}, it has been possible to prove the existence of minimizers in several cases. Le Bris \cite{LeBris1993} proved that for a neutral or positively charged system, the Kohn-Sham problem with LDA exchange-correlation energy admits a minimizer. A similar result was proved by Anantharaman and Canc\`es \cite{Anantharaman2009} for the so-called extended-Kohn-Sham model with LDA exchange-correlation energy. \\
%
%
%
%
%
\indent The purpose of the present article is to extend the result by Anantharaman and Canc\`es to spin-polarized systems, the electrons of the molecular system into consideration being subjected to the electric potential $V$ created by the nuclei, and to an arbitrary external magnetic field $\bB$ that vanishes at infinity. In order to take into account spin effects, we have to resort to spin density functional theory (SDFT). In this theory, all magnetic contributions coming from orbital magnetism (paramagnetic current, spin-orbit coupling,...) are  neglected. Historically, while Kohn and Sham briefly discussed the inclusion of spin effects in their model, the general theory was pioneered by von Barth and Hedin \cite{Barth1972} and is known as the local spin density approximation (LSDA). These authors proposed the following ansatz to transform a spin-unpolarized exchange-correlation energy to a spin-polarized version of it:
\[
	E_{\xc}^{\LSDA}(\rho^+, \rho^-) := \dfrac12 \left[ E_{\xc}^{\LDA} (2 \rho^+) + E^{\LDA}_{\xc} (2 \rho^-) \right],
\]
where $E_{\xc}^{\LDA}$ is the spinless exchange-correlation energy, and $\rho^{+/-}$ are the eigenvalues of the $2 \times 2$ spin density matrix (see Sec. 2 for details).
There are two other major differences between spin-polarized and spin-unpolarized models. First, the ground state of spin-unpolarized models is given by a minimization problem onto the set of electronic densities, while in spin-polarized models, it is given by a minimization problem onto the set of spin density matrices, consisting of $2 \times 2$ hermitian matrices. Finally, the magnetic field adds a Zeeman-type term~$- \mu \int \bB \cdot \bm$ to the energy functional, where $\bm$ is the spin angular momentum density. \\
\indent Due to all those additional difficulties with respect to the spinless case, the fully polarized SDFT has not been very popular until recently. Chemists generally prefer its collinear version (collinear-SDFT), where all the spins are constrained to be orientated along a fixed direction on the whole space. This allows one to work with two scalar fields (one for spin-up, and one for spin-down), instead of fields of hermitian matrices. While this simplification provides very good results, it misses some physical properties (spin dynamics \cite{Sharma2007}, frustrated solids \cite{Bulik2013}, ...). The implementation of the unconstrained (fully polarizable) model appeared with the work of Sandratskii and Guletskii \cite{Sandratskii1986}, and K\"ubler \textit{et al.} \cite{Kubler1988_1,Kubler1988}, and this model is becoming a standard tool nowadays. To the best of our knowledge, no rigorous proof of the existence of solutions has yet been provided for this case.\\
\indent Our result is that, under the same hypotheses as in \cite{Anantharaman2009}, plus some mild conditions on $\bB$, the existence of minimizers is still ensured for neutral or positively charged systems. Whereas the main tools of the proof are similar to those used in \cite{Anantharaman2009}, namely concentration compactness techniques, some adaptations are necessary, in particular to handle the Zeeman term. The structure of the article is as follows. We first recall how to derive the LSDA models, and formulate the main theorem. Then, we break the proof of the theorem into several lemmas, that we prove at the end of the paper.


\section{Derivation of the local spin density approximation models}

We recall how the extended Kohn-Sham models are derived in the spin setting. We start from the Schr\"odinger-Pauli Hamiltonian for $N$-electrons in the Born-Oppenheimer approximation. In atomic units, this operator reads
\[
	H^{\textrm{SP}} (V, \bA) = \sum_{i=1}^N  \dfrac12 \left( - \ri \nabla_i + \bA(\br_i) \right)^2 \bI_2 + \sum_{i=1}^N  V(\br_i)  \bI_2 - \mu \sum_{i=1}^N  \bB(\br_i) \cdot \bsigma_i    + \sum_{1 \le i < j \le N} \dfrac{1}{| \br_i - \br_j |} \bI_2 ,
\]
where $\bI_2$ is the $2 \times 2$ identity matrix,
\begin{equation} \label{V}
	V(\br) = -\sum_{k=1}^M \dfrac{z_k}{| \br - \bR_k| }
\end{equation}
is the electric potential generated by the nuclei,  $\bA$ is the external magnetic vector potential, and $\bB := \nabla \times \bA$ is the external magnetic field. We denote by $\br_i$ (resp. $\bR_k$) the positions of the electrons (resp. nuclei). The charge of the $k$-th nucleus is $z_k \in \N^*$ and $Z := \sum_{k=1}^M z_k$ is the total nuclear charge. We can assume without loss of generality that $\bR_1 = \bnull$. The constant $\mu$ is the Bohr magneton. Its value is $1/2$ in atomic units, but we prefer to keep the notation $\mu$ in the rest of the paper. The term $\bsigma_i$ appearing in the Hamiltonian contains the Pauli matrices acting on the $i$-th spin variable:
\[
	\bsigma_i := \left( \sigma_{xi}, \sigma_{yi}, \sigma_{zi} \right) = \left( \begin{pmatrix} 0 & 1 \\ 1 & 0 \end{pmatrix}_i , \begin{pmatrix} 0 & -\ri \\ \ri & 0 \end{pmatrix}_i, \begin{pmatrix}  1 & 0 \\ 0 & -1 \end{pmatrix}_i     \right).
\]

\indent Although the magnetic field $\bB$ and magnetic vector potential $\bA$ are linked by the relation $\bB = \nabla \times \bA$, it is often preferable to consider them as two independent fields. Indeed, $\bB$ acts on the spin of the electrons, while $\bA$ acts on the spatial component of the spin-orbitals. For instance, would we be interested only in studying orbital effects (\textit{e.g.} paramagnetic current), we would neglect the spin effects. We would then take $\bB = \bnull$ and $\bA \neq \bnull$. Such an approximation leads to the so-called current-density functional theory \cite{Vignale1988}. In this article, we are interested in spin effects. We therefore set $\bA = \bnull$, which amounts to neglecting the paramagnetic currents, while keeping $\bB \neq \bnull$. With this approximation, our Hamiltonian for $N$ electrons reads
\[
	H(V, \bB) = \sum_{i=1}^N -\frac12 \Delta_i + \sum_{i=1}^N V(\br_i) - \mu \sum_{i=1}^N \bB(\br_i) \cdot \bsigma_i + \sum_{1
\le i < j \le N} \dfrac{1}{| \br_i - \br_j|} .
\]
This Hamiltonian acts on the fermionic Hilbert space
\begin{align*}
	\bigwedge_{i=1}^N L^2 (\R^3, \C^2) := & \Big\{ \Psi(\br_1, s_1, \cdots , \br_N, s_n), 
		\ \br_i \in \R^3, s_i \in \{ \spinup, \spindown\}, \\
	& \quad		\sum_{s_1, \cdots s_N \in \{ \spinup , \spindown\}^N} \int_{\R^{3N}} \| \Psi(\br_1, s_1, \cdots ) \|^2 \rd^3 \br_1 \cdots  \rd^3 \br_n < \infty,  \\
		& \quad \forall p \in \fs_N, \quad \Psi(\br_{p(1)}, s_{p(1)}, \cdots ) = \epsilon(p) \Psi(\br_1, s_1, \cdots) \Big\},
\end{align*}
where $\epsilon(p)$ is the parity of the permutation $p$, endowed with the scalar product
\[
	\bra \Psi_1 | \Psi_2 \ket = \sum_{(s_1, \cdots s_N) \in \{ \spinup , \spindown \}^N} \int_{\R^{3N}} \overline{\Psi_1(\br_1, s_1, \cdots)} \Psi_2(\br_1, s_1, \cdots) \rd^3 \br_1 \cdots \rd^3 \br_N.
\]
Its form domain $\bigwedge_{i=1}^N H^1(\R^3, \C^2)$ is defined similarly.

The ground state energy of the system is obtained by solving the minimization problem
\[
	E(V, \bB) := \inf \left\{ \bra \Psi | H(V, \bB) | \Psi \ket, \quad \Psi \in \bigwedge_{i=1}^N H^1 (\R^3, \C^2), \quad \| \Psi \|_{L^2} = 1 \right\}.
\]
In order to convexify the problem, we introduce, for a wave function $\Psi~\in~\bigwedge_{i=1}^N H^1 (\R^3, \C^2)$ satisfying $\| \Psi \|=1$, the $N$-body density matrix
\[
	\Gamma_\Psi := | \Psi \ket \ \bra \Psi |.
\]
The minimization problem can be recast as
\[
	E(V, \bB) = \inf \left\{ \Tr \left(  H(V, \bB) \Gamma \right), \quad \Gamma \in \cW_N \right\}
\]
where $\cW_N$ is the set of pure state $N$-body density matrices defined by
\[
	\cW_N := \left\{ \Gamma_\Psi, \quad \Psi \in \bigwedge_{i=1}^N H^1 (\R^3, \C^2), \quad \| \Psi \|_{L^2} = 1 \right\}.
\]
In this article, we study the extended-Kohn-Sham model based on mixed state $N$-body density matrices, for this problem has better properties mathematically speaking, and allows one to handle more general physical situations as, for instance, positive temperatures. The set $\cM_N$  of mixed state $N$-body density matrices is defined as the convex hull of $\cW_N$. The minimization problem for mixed states reads
\[
	E(V, \bB) := \inf \left\{ \Tr \left(  H(V, \bB) \Gamma \right), \quad \Gamma \in \cM_N \right\}.
\]
Then, for $\Gamma \in \cM_N$, direct calculations lead to
\begin{equation} \label{recast}
	\fl \Tr \left( H(V, \bB) \Gamma \right)= \Tr \left( H(0, \bnull) \Gamma \right) + 
		\int_{\R^3} \tr_{\C^2} \left[ 
		\begin{pmatrix} V - \mu B_z & - \mu B_x + \ri \mu B_y \\ -\mu  B_x - \ri \mu B_y & V + \mu B_z \end{pmatrix}
			\begin{pmatrix} \rho^{\spinup \spinup}_\Gamma & \rho^{\spinup \spindown}_\Gamma \\  \rho^{\spindown \spinup}_\Gamma &  \rho^{\spindown \spindown}_\Gamma \end{pmatrix}
			\right],
\end{equation}
where $\Gamma(\br_1, s_1, \cdots ; \br_1', s_1', \cdots)$ denotes the kernel of $\Gamma$, and, for $\alpha, \beta \in \{ \spinup ,\spindown\}^2$, 
\[
	\rho_\Gamma^{\alpha \beta} (\br) := N \sum_{(s_2, \cdots, s_N) \in \{ \spinup , \spindown\}^{N-1}} \int_{\R^{3(N-1)}} 
		\Gamma(\br, \alpha, \br_2, s_2, \cdots ; \br, \beta, \br_2, s_2, \cdots) \rd^3 \br_2 \cdots \rd^3 \br_N.
\]
In the following, we write
\[
	U := \begin{pmatrix} V - \mu B_z & -\mu B_x + \ri \mu B_y \\ -\mu B_x - \ri \mu B_y & V + \mu B_z \end{pmatrix} \quad \text{and} \quad 
	R_\Gamma := \begin{pmatrix} \rho^{\spinup \spinup}_\Gamma & \rho^{\spinup \spindown}_\Gamma \\  \rho^{\spindown \spinup}_\Gamma &  \rho^{\spindown \spindown}_\Gamma \end{pmatrix} .
\]
This last $2 \times 2$ matrix is called the spin density matrix. Note that when $\bB = \bnull$, one recovers the usual potential energy density $V \rho_\Gamma$ appearing in spin-unpolarized DFT.
Introducing the spin angular momentum density $\bm_\Gamma = \tr_{\C^2} \left[ \bsigma \cdot R_\Gamma\right]$, and the total electronic density $\rho_{\Gamma} = \rho_\Gamma^{\spinup \spinup} + \rho_\Gamma^{\spindown \spindown}$, it holds

\begin{equation} \label{bm}
	\tr_{\C^2} \left[ U R_\Gamma \right] = V \rho_\Gamma - \mu \bB \cdot \bm_\Gamma .
\end{equation}
%

We now apply the constrained search method introduced and studied by Levy \cite{Levy1979}, Valone \cite{Valone1980} and Lieb \cite{Lieb1983}, and write the minimization problem (\ref{recast}) in terms of $R_\Gamma$:
\begin{equation} \label{DFT}
	E(V, \bB) = \inf \left\{ F(R) + \int_{\R^3} \tr_{\C^2} \left[ U R \right], \quad R \in \cJ_N \right\},
\end{equation}
with
\[
	 F(R) := \inf \left\{ \Tr \left[ H(0, \bnull) \Gamma \right], \quad \Gamma \in \cM_N, \quad R_\Gamma = R \right\}.
\]
The set $\cJ_N$ is defined as
\[
	\cJ_N := \left\{ R \in \cM_{2 \times 2}(L^1(\R^3)), \quad \exists \Gamma \in \cM_N, \quad R_\Gamma = R \right\},
\]
where $ \cM_{2 \times 2}(L^1(\R^3))$ is the space of $2 \times 2$ matrices with entries in $L^1(\R^3)$. We recently proved \cite{Gontier2013} the following characterization for $\cJ_N$ in the mixed state setting:
\[
	\hspace{-2em} \cJ_N = \left\{ R \in \cM_{2 \times 2} (L^1(\R^3)), \quad R^* = R, \quad R \ge 0, \quad \int_{\R^3} \tr_{\C^2} \left[ R \right] = N, \quad \sqrt{R} \in \cM_{2 \times 2} \left( H^1(\R^3) \right) \right\} .
\]
\indent As mentioned before, the functional $F$ cannot be straightforwardly evaluated. In order to make this problem practical, we approximate $F$. It is standard since the work of Kohn and Sham \cite{Kohn1965} to approximate this functional by studying a system of non-interacting electrons. For this purpose, we introduce, for a mixed state $\Gamma \in \cM_N$, the $1$-body density matrix
\[
	\gamma_\Gamma := \begin{pmatrix} \gamma_\Gamma^{\spinup \spinup} & \gamma_\Gamma^{\spinup \spindown} \\
		\gamma_\Gamma^{\spindown \spinup} & \gamma_\Gamma^{\spindown \spindown} \end{pmatrix}
\]
where
\[
	\gamma_\Gamma^{\alpha \beta}(\br,\br') := N \sum_{(s_2, s_3, \cdots) \in \{ \spinup , \spindown\}^{N-1}} \int_{\R^{3(N-1)}} 
		\Gamma(\br, \alpha, \br_2, s_2, \cdots ; \br', \beta, \br_2, s_2, \cdots ) \rd^3 \br_2 \cdots \rd^3 \br_N .
\]
The set of mixed-state 1-body density matrices is
\[
	\cP_N := \{ \gamma_\Gamma, \quad \Gamma \in \cM_N \},
\]
and, identifying the kernel $\gamma(\br, \br')$ with the corresponding operator of $\mathcal{S}(L^2(\R^3, \C^2))$, the space of self-adjoint operators on $L^2(\R^3, \C^2)$, Coleman \cite{Coleman1963} proved that
\[
	\cP_N = \left\{ \gamma \in \mathcal{S}(L^2(\R^3, \C^2)), \quad 0 \le \gamma \le 1, \quad \Tr( \gamma) = N, \quad \Tr( - \Delta \gamma) < \infty \right\}.
\]
In a similar way, we can define, for $\lambda > 0$,
\begin{equation} \label{Plambda}
	\cP_\lambda := \left\{ \gamma \in \mathcal{S}(L^2(\R^3, \C^2)), \quad 0 \le \gamma \le 1, \quad \Tr( \gamma) = \lambda, \quad \Tr( - \Delta \gamma) < \infty \right\}.
\end{equation}
A more practical and equivalent formulation of the Coleman result is that, using the spectral theory for compact self-adjoint operators, we can write the components $\gamma^{\alpha \beta}$ of any $\gamma \in \cP_\lambda$ in the form
\begin{align} \label{Coleman}
	\fl \gamma^{\alpha \beta} (\br, \br') & = \sum_{k=1}^\infty n_k \phi_k^\alpha(\br) \overline{\phi_k^\beta(\br')}, 
	 \ 0 \le n_k \le 1, \ \sum_{k=1}^\infty n_k = \lambda, \ \Phi_k = \begin{pmatrix} \phi_k^\spinup \\ \phi_k^\spindown \end{pmatrix} \in L^2(\R^3, \C^2),  \bra \Phi_k | \Phi_l \ket = \delta_{kl}, \nonumber\\
	 & 
	   \Tr( - \Delta \gamma) := \sum_{k=1}^\infty n_k \| \nabla \Phi_k \|_{L^2}^2 = \Tr (- \Delta \gamma^{\spinup \spinup}) + \Tr(- \Delta \gamma^{\spindown \spindown}) < \infty.
\end{align}
Notice that $\gamma_\Gamma ( \br, \br) = R_\Gamma( \br)$, so that we will write $R_\gamma(\br) := \gamma (\br, \br)$ for $\gamma \in \cP_N$. We finally introduce
\[
	J_\lambda	 := \left\{ R \in \cM_{2 \times 2}(L^1(\R^3)), \quad \exists \gamma \in \cP_\lambda, \quad R = R_\gamma \right\}.
\]

\indent The extended version of the Kohn-Sham approach consists then in splitting the unknown functional $F(R)$ into three parts:
\begin{equation} \label{KS_split}
	F(R) =  T_{\textrm{KS}}(R)+ J(R)+ E_{\xc}(R).
\end{equation}
The first term $T_{\textrm{KS}}$ represents the kinetic energy of a non-interacting electronic system. It reads, in the one-body formalism,
\[
	\forall R \in \cJ_\lambda, \quad T_{\textrm{KS}}(R) := \inf \left\{ \dfrac12 \Tr \left( - \Delta \gamma \right) , \quad \gamma \in \cP_\lambda, \quad R_\gamma = R \right\}.
\]
The second term is the Hartree term, defined by
\[
	J(\rho) := \dfrac12 \iint_{\R^3 \times \R^3} \dfrac{\rho(\br) \rho(\br')}{| \br - \br'|} \rd^3 \br \ \rd^3 \br' .
\]
Finally, the last term is the exchange-correlation functional defined by
\[
	\quad E_{\xc}(R)  := F(R) - T_{\textrm{KS}}(R) - J(R).
\]
Notice that because $F$ is a non-explicit functional, $E_{\xc}$ is also a non-explicit functional. However, the purpose of splitting $F$ according to (\ref{KS_split}) is that $E_\xc$ is an order of magnitude smaller that $F$. We can gain another order of magnitude in accuracy with respect to the reduced Hartree-Fock model \cite{Solovej1991} (where $E_{\xc} = 0$) with a good approximation of the functional $E_{\xc}$.\\
\indent The local-spin density approximation introduced by von Barth and Hedin \cite{Barth1972} consists in writing
\begin{equation} \label{vonBarth}
	E_{\xc}(R) \approx E_{\xc}^{\LSDA}(\rho^+, \rho^-) := \dfrac12 \left[ E_{\xc}^{\LDA} (2 \rho^+) + E^{\LDA}_{\xc} (2 \rho^-) \right]
\end{equation}
where $\rho^{+/-}$ are the two eigenvalues of the $2 \times 2$ matrix $R$, and $E_{\xc}^{\LDA}$ is the standard exchange-correlation functional in the non-polarized case, that we can write under the form \cite{Kohn1965}
\begin{equation} \label{g}
	E_{\xc}^\LDA (\rho) = \int_{\R^3} g(\rho (\br)) \rd^3 \br .
\end{equation}
We emphasize that the polarization rule (\ref{vonBarth}) is exact for the exchange part of the exchange-correlation energy, and that von Barth and Hedin proposed to use the same formula for the correlation part. The fact that $E_{\xc}^\LSDA$ only depends on  $R$ via its eigenvalues comes from the locality of the functional. Indeed, this energy functional must be invariant with respect to local spin rotations. Because $R$ is hermitian at each point, we can always diagonalize $R$ locally, so that a local energy can only depend on the two eigenvalues of $R$.\\

In this article, we will deal with exchange-correlation functionals of the form (\ref{vonBarth})-(\ref{g}). The mathematical properties of the standard LDA exchange-correlation functional are similar \cite{Perdew1981} to the one of the $X\alpha$-functional introduced by Slater \cite{Slater1951}
\[
	E_{\xc}^{\LDA, X\alpha} (\rho) = - C_X \int_{\R^3} \rho^{4/3}(\br) \rd^3 \br.
\]

Altogether, by recasting problem (\ref{DFT}) in terms of the one-body density matrices, we end up with a variational problem of the form
\begin{equation} \label{I_lambda}
	I_\lambda := \inf \left\{ \cE(\gamma), \quad \gamma \in \cP_\lambda \right\}
\end{equation}
where
\[
	\cE(\gamma) = \dfrac12 \Tr \left( - \Delta \gamma^{\spinup \spinup} \right) +  \dfrac12 \Tr \left( - \Delta \gamma^{\spindown \spindown} \right) + J(\rho_\gamma) + \int_{\R^3} \tr_{\C^2} \left[ U R_\gamma  \right] \rd^3 \br +E_{\xc}^\LSDA(\rho_\gamma^+, \rho_\gamma^-)
\]
and where $P_\lambda$ has been defined in (\ref{Plambda}). The physical situation corresponds to $\lambda = N \in \N$, but as usual in variational problems set on the whole space, it is useful to relax the constraint $\Tr(\gamma) = N$ to allow the particles to escape to infinity.

We can recover some other common models by further constraining the minimization set. For instance, the collinear-SDFT consists in minimizing the functional $\cE$ onto the set
\[
	\cP_\lambda^{\textrm{collinear}} := \left\{ \gamma \in \cP_\lambda, \quad \gamma^{\spinup \spindown} = \gamma^{\spindown \spinup} = 0 \right\}.
\]
In this case, the matrices $\gamma$ and $R$ are both diagonal. In particular, the two eigenvalues of $R$ are $\{ \rho^+, \rho^- \} = \{ \rho^{\spinup \spinup}, \rho^{\spindown \spindown} \}$. In this model, it holds that
\[
	\int_{\R^3} \tr_{\C^2} \left[ U R \right] = \int_{\R^3} V (\rho^{\spinup \spinup} + \rho^{\spindown \spindown}) - \mu \int_{\R^3} B_z (\rho^{\spinup \spinup} - \rho^{\spindown \spindown}) = \int_{\R^3} V \rho - \mu \int_{\R^3} B_z \rho\  \zeta.
\]
where
\[
	\zeta := \dfrac{\rho^{\spinup \spinup} - \rho^{\spindown \spindown}}{\rho^{\spinup \spinup} + \rho^{\spindown \spindown}} \in [-1, 1]
\]
is the relative spin-polarization. This model is therefore simpler than the non-collinear spin-polarized model, as we are not dealing with fields of matrices, but with two scalar fields. Physically, it corresponds to constraining the spin along a fixed direction on the whole space. This method provides results in good agreement with experiments whenever the energy accounting for the non-collinearity of the spins is negligible.


Then, the unpolarized case consists in minimizing the functional $\cE$ onto the set
\[
	\cP_\lambda^{\textrm{unpolarized}} := \left\{ \gamma \in \cP_\lambda, \quad \gamma^{\spinup \spindown} = \gamma^{\spindown \spinup} = 0, \quad \gamma^{\spinup \spinup} = \gamma^{\spindown \spindown} \right\}.
\]
Equivalently, it corresponds to the collinear case with $\zeta \equiv 0$. It then holds that
\[
	\int_{\R^3} \tr_{\C^2} \left[ U R \right] = \int_{\R^3} V \rho ,
\]
so that the model is independent of the magnetic field $\bB$, and can be used whenever spin effects are negligible.  We refer to \cite{Anantharaman2009} for a mathematical introduction of this model.


\section{An existence result for the Kohn-Sham LSDA model}

The main result of this article is the following
\begin{theorem}
	Under the following assumptions
	\begin{itemize}
		\item[1/]  the function $g$ in (\ref{g}) is of class $\cC^1(\R^+)$ and satisfies:
		\begin{equation} \label{cond_g}
			\left\{ \begin{aligned}
			& g(0) = 0 \\
			& g' \le 0 \\
			& \exists \ 0 < \beta^- \le \beta^+ < \frac23, \quad \sup_{\rho \in \R^+} \dfrac{|g'(\rho)|}{\rho^{\beta^-} + \rho^{\beta^+}} < \infty\\
			& \exists \ 1 \le \alpha < \dfrac32, \quad \limsup_{\rho \to 0^+} \dfrac{g(\rho)}{\rho^{\alpha}} < 0,
			\end{aligned}
			\right.
		\end{equation}
		\item[2/] all entries of $U$ are in $L^{\frac32+\epsilon}(\R^3) + L^{\infty}(\R^3)$ and vanish at infinity, and $V := \tr_{\C^2}(U)$ has the form (\ref{V}),
	\end{itemize}
	the problem $I_\lambda$ defined in (\ref{I_lambda}) has a minimizer whenever $\lambda \le Z$.
\end{theorem}

\noindent \textbf{Remark 1.} The assumptions (\ref{cond_g}) are the same as in \cite{Anantharaman2009}. What is added in this article is the introduction of a magnetic field. 

\noindent \textbf{Remark 2.} This result does not make any assumption on the strength of the magnetic field $\bB$ other than that it vanishes at infinity. If $\bB$ becomes infinite at infinity, it is easy to see that the energy is not bounded below: we can orientate the spins of all electrons along the magnetic field and push them to infinity, so that the energy can be arbitrarily negative. \\



\noindent\textit{Proof of Theorem 1:}\\
We use the concentration-compactness method introduced in \cite{Lions1984}. We therefore introduce the problem at infinity 
\[
	I_\lambda^\infty = \inf \left\{ \cE^\infty(\gamma), \quad \gamma \in \cP_\lambda \right\},
\]
where
\[
	\cE^\infty(\gamma) := \dfrac12 \Tr \left( - \Delta \gamma^{\spinup \spinup} \right) +  \dfrac12 \Tr \left( - \Delta \gamma^{\spindown \spindown} \right) + J(\rho) + E_{\xc}^\LSDA(\rho^+, \rho^-).
\]

We will need several lemmas, the proofs of which are postponed until the following section for the sake of clarity. We begin with some functional inequalities:


\begin{lemma} \label{lemma1}
	There exists a constant $C$ such that for all $\lambda > 0$ and all $\gamma \in \cP_\lambda$, it holds
	\[
		\| \nabla R_\gamma \|_{L^{3/2}} \le C \Tr( - \Delta \gamma)
		 \quad \text{and} \quad \| \nabla \rho_\gamma^{+/-} \|_{L^{3/2}} \le C \Tr( - \Delta \gamma).
	\]
	In particular, for all $1 \le p \le 3$, there exists $C_p$ such that, for all $\lambda > 0$ and all $\gamma \in \cP_\lambda$,
	\begin{equation} \label{eq:lemma1}
		\| R_\gamma \|_{L^p} \le C_p \lambda^{\frac{3-p}{2p}} \Tr (- \Delta \gamma)^{\frac{3(p-1)}{2p}} ,
	\end{equation}
	and similarly for $\rho_\gamma^{+/-}$.
\end{lemma}
We easily deduce from the above lemma that the energies $I_\lambda$ and $I_\lambda^\infty$ are bounded below:
\begin{lemma} \label{lemma2}
	For all $\lambda > 0$, we have $I_\lambda > -\infty$ and $I_\lambda^\infty > - \infty$. Moreover, all minimizing sequences $(\gamma_n)$ for $I_\lambda$ or $I_\lambda^\infty$ are bounded in the Banach space $\cB$, where
\[
	\cB := \{ \gamma \in \mathcal{S}(L^2(\R^3, \C^2)), \quad \| \gamma \|_\cB := \Tr( | \gamma |) + \Tr (| |\nabla | \gamma | \nabla | |)  < \infty \} .
\]
\end{lemma}

In the following, we consider sequences  $(\gamma_n)_{n \in \N^*} \in \cB$, and we will write $R_n := R_{\gamma_n}$ and $\rho_n := \rho_{\gamma_n}$.

\begin{lemma}  \label{lemma3}
	Let $(\gamma_n)_{n \in \N^*}$ be a bounded sequence of $\cB$. Then, there exists $\gamma_0 \in \cB$, such that,  up to a subsequence, $\gamma_n$  converges to $\gamma_0$ for the weak-$*$ topology of $\cB$, all components of $R_n$ converge to their respective components in $R_0$ strongly in $L^p_{\loc}(\R^3)$  for $1 \le p < 3$, weakly in $L^p(\R^3)$ for $1 \le p \le 3$, and almost everywhere. The eigenvalues of $R_n$ converge to the eigenvalues of $R_0$ strongly in $L^p_{\loc}(\R^3)$ for $1 \le p < 3$, weakly in $L^p(\R^3)$ for $1 \le p \le 3$ and almost everywhere.  \\
	\indent Moreover, if $\gamma_n \in \cP_\lambda$ for all $n$, and $\gamma_0 \in \cP_\lambda$, the convergences hold strongly in $L^p(\R^3)$ for $1 \le p < 3$, and $\cE(\gamma_0) \le \liminf \cE( \gamma_n)$.
\end{lemma}

It follows from Lemma \ref{lemma2} and Lemma \ref{lemma3} that one can extract from any minimizing sequence $(\gamma_n)_{n \in \N^*}$ of (\ref{I_lambda}) a minimizing sequence, still denoted by $(\gamma_n)$, converging to some $\gamma_0$ for the weak-$*$ topology of $\cB$.
In particular, $0 \le \gamma_0 \le 1$ and $\Tr(-\Delta \gamma_0) < \infty$. 
To prove that $\gamma_0$ is indeed a minimizer of (\ref{I_lambda}), it remains to prove that $\Tr(\gamma_0) = \lambda$. Let $\alpha = \Tr(\gamma_0)$.
It is easy to get $\alpha \le \lambda$. If $\alpha < \lambda$, then we have loss of compactness (some electrons leak away). Therefore, to prove that $\alpha = \lambda$ (at least when $\lambda \le Z$), 
we need to control the behavior at infinity of the minimizers, which is not as simple as in \cite{Anantharaman2009} because of the Zeeman term $- \mu \int \bB \cdot \bm$. 
In order to control this term, we introduce the following "flip" transformation:
\begin{align} \label{transformation}
	\text{for } &  \Phi = \begin{pmatrix} \phi^\spinup \\ \phi^\spindown \end{pmatrix}, \  \text{we define }  \widetilde \Phi :=  \begin{pmatrix} \overline{\phi^\spindown} \\ - \overline{\phi^\spinup} \end{pmatrix}, \nonumber \\
	\text{for } & 
	\gamma = \sum n_k | \Phi_k \ket \bra \Phi_k | , \ \text{we define }  \widetilde \gamma := \sum n_k | \widetilde \Phi_k \ket \ \bra \widetilde \Phi_k |.
\end{align}
Note that if
\[
	 \gamma = \begin{pmatrix} 
		\gamma^{\spinup \spinup} & \gamma^{\spinup \spindown} \\
		\gamma^{\spindown \spinup} & \gamma^{\spindown \spindown} 
		\end{pmatrix} 
	\quad \text{and} \quad
		R_\gamma = \begin{pmatrix} 
		R^{\spinup \spinup} & R^{\spinup \spindown} \\
		R^{\spindown \spinup} & R^{\spindown \spindown} 
		\end{pmatrix}, 
\]
then
\[
		\widetilde \gamma (\bx, \by)= \begin{pmatrix} 
			{\gamma^{\spindown \spindown}} & - \gamma^{\spinup \spindown} \\
				- \gamma^{\spindown \spinup} & {\gamma^{\spinup \spinup}}
				\end{pmatrix}(\by, \bx)
			\quad \text{and} \quad
		R_{\widetilde \gamma} = \begin{pmatrix} 
		R^{\spindown \spindown} & - R^{\spinup \spindown} \\
		- R^{\spindown \spinup} & R^{\spinup \spinup} 
		\end{pmatrix}, 
\]
from which we deduce the following lemma, whose proof is straightforward.
\begin{lemma}  \label{lemma4}
	If $\gamma \in \cP_\lambda$, then $\widetilde \gamma \in \cP_\lambda$. Moreover, it holds that
		$ \Tr( - \Delta \widetilde \gamma_n) = \Tr( - \Delta \gamma_n) $, 
		 $ \widetilde \rho = \rho$, and 
		 $\widetilde \bm = - \bm$, where $\rho$ and $\bm$ have been defined in (\ref{bm}).
	In particular, it holds that
\begin{equation} \label{ControlZeeman}
	\tr_{\C^2} \left[ U R  \right] + \tr_{\C^2} \left[ U \widetilde R \right] = 2 \int_{\R^3} V \rho .
\end{equation}
\end{lemma}
In other words, this transformation flips the spin-up and spin-down channels. This lemma allows to cancel the Zeeman term, and is an essential tool throughout the proof. We can first prove


\begin{lemma}  \label{lemma5}
	~\begin{enumerate}[label=(\roman*)]
		\item For all  $\lambda > 0$, $I_\lambda^\infty < 0$.
		\item For all $\lambda > 0$, $-\infty < I_\lambda < I_\lambda^\infty < 0$.
		\item For all $0 < \mu < \lambda$, $I_\lambda \le I_\mu + I^\infty_{\lambda - \mu}$.
		\item The functions $\lambda \mapsto I_\lambda$ and $\lambda \mapsto I^\infty_\lambda$ are non increasing.
	\end{enumerate}
\end{lemma}
We then have the important result

\begin{lemma}  \label{lemma6}
	Let $\lambda > 0$ and $(\gamma_n)_{n \in \N^*} \in \cP_\lambda$ be any minimizing sequence of $I_\lambda$ that converges to some $\gamma_0$ for the weak-$*$ topology of $\cB$. Let $\alpha := \Tr( \gamma_0)$. Then
	\begin{enumerate}[label=(\roman*)]
		\item $\alpha \le \lambda$.
		\item $\alpha \neq 0$.
		\item If $0 < \alpha < \lambda$, then $\gamma_0$ is a minimizer for the problem $I_\alpha$, there exists $\beta > 0$ with $\alpha + \beta \le \lambda$ such that $I_\beta^\infty$ has also a minimizer, and $I_{\lambda} = I_\alpha + I_\beta^\infty + I_{\lambda - \alpha - \beta}^\infty$.
	\end{enumerate}
\end{lemma}

According to this lemma, if $\alpha < \lambda$, $\gamma_0$ is a minimizer for $I_\alpha$. In this case, it satisfies the Euler-Lagrange equation
\[
	\gamma_0 = \b1 (H_{\gamma_0} < \epsilon_F) + \delta \quad \text{with} \quad \delta \subset \textrm{Ker}( H_{\gamma_0} - \epsilon_F)
\]
with $H_{\gamma_0}$ as defined in (\ref{Hgamma}). We then use the very general
\begin{lemma} \label{lemma7}
	It holds $\sigma_{\textrm{ess}} (H_{\gamma_0}) = [0, +\infty[$. Moreover, if $0 < \lambda < Z$, then $H_{\gamma_0}$ has infinitely many negative eigenvalues, and every eigenvector corresponding to such an eigenvalue is exponentially decreasing.
\end{lemma}

From this lemma, we deduce the concentration-compactness result:
\begin{lemma}  \label{lemma8}
	Let $0 < \alpha, \beta$ be such that $\alpha + \beta \le Z$. Suppose that $I_\alpha$ and $I_{\beta}^\infty$ admit minimizers. Then
	\[
		I_{\alpha + \beta} < I_\alpha + I_\beta^\infty \quad (\  < I_\alpha) .
	\]
\end{lemma}
The end of the proof goes as follows. Let us suppose that $\lambda \le Z$, and $\alpha < \lambda$. Then, according to Lemma \ref{lemma6}, $\gamma_0$ is a minimizer for $I_\alpha$, and there exists $\beta > 0$ such that $\alpha + \beta \le \lambda \le Z$ so that $I_\beta^\infty$ has also a minimizer, and it holds $I_\lambda = I_{\alpha} + I_{\beta}^\infty + I_{\lambda - \alpha - \beta}^\infty$. Moreover, Lemma \ref{lemma8} holds, and $I_{\alpha + \beta} < I_{\alpha}  + I_{\beta}^\infty$. Finally, we get
\[
	I_\lambda = I_{\alpha} + I_{\beta}^\infty + I_{\lambda - \alpha - \beta}^\infty > I_{\alpha + \beta} + I^\infty_{\lambda - \alpha - \beta} ,
\]
which contradicts the third point of Lemma \ref{lemma5}. \\
Therefore, it holds $\alpha = \lambda$, and, according to Lemma \ref{lemma3}, $\gamma_0$ is a minimizer for $I_\lambda$, which concludes the proof.


\section{Proofs of the lemmas}

\begin{proof}[Proof of lemma \ref{lemma1}]
	Let $\lambda > 0$ and $\gamma \in \cP_\lambda$. We use the representation (\ref{Coleman}) of $\gamma$, and write
\begin{equation*}
	\begin{aligned}
	\gamma^{\alpha \beta} (\br, \br') & = \sum_{k=1}^\infty n_k \phi_k^\alpha(\br) \overline{\phi_k^\beta(\br')}, 
	 \ 0 \le n_k \le 1, \ \sum_{k=1}^\infty n_k = \lambda, \\ 
	 & \Phi_k = \begin{pmatrix} \phi_k^\spinup \\ \phi_k^\spindown \end{pmatrix} \in L^2(\R^3, \C^2),  \ 
	  \bra \Phi_k | \Phi_l \ket = \delta_{kl}, \ \Tr( - \Delta \gamma) := \sum_{k=1}^\infty n_k \| \nabla \Phi_k \|_{L^2}^2 < \infty.
	 \end{aligned}
\end{equation*}
In particular, $ \rho^{\alpha \beta} (\br) = \sum n_k \phi_k^\alpha(\br) \overline{\phi_k^\beta(\br)}$. Differentiating this expression, and using the Cauchy-Schwarz inequality, it holds
\begin{align*}
	| \nabla \rho^{\alpha \beta} |^2 & = \left| \sum_{k=1}^\infty n_k \left( \nabla \phi_k^\alpha(\br) \overline{\phi_k^\beta(\br)} +  \phi_k^\alpha(\br) \nabla \overline{\phi_k^\beta(\br)} \right) \right|^2 \\
	& \le \left| \sum_{k=1}^\infty n_k \left( | \nabla \phi_k^\alpha |^2 + | \nabla \phi_k^\beta |^2 \right )^{1/2} \left( | \phi_k^\alpha |^2 + | \phi_k^\beta |^2 \right )^{1/2} \right|^2 \\
	& \le \left[ \sum_{k=1}^\infty n_k \left( | \nabla \phi_k^\alpha |^2 + | \nabla \phi_k^\beta |^2 \right ) \right]
		\left[ \sum_{k=1}^\infty n_k \left( | \phi_k^\alpha |^2 + | \phi_k^\beta |^2 \right ) \right] .
\end{align*}

We let $\tau^\alpha := \sum_{k=1}^\infty n_k | \nabla \phi_k^\alpha |^2$, so that $\tau^\alpha \in L^1(\R^3)$ and $\int_{\R^3} \tau^\alpha = \Tr( - \Delta \gamma^{\alpha \alpha})$. The previous inequality leads to the point-wise estimate
\begin{equation} \label{nabla_rho}
	| \nabla \rho^{\alpha \beta} |  \le \left( \tau^\alpha + \tau^\beta \right)^{1/2} \left( \rho^{\alpha \alpha} + \rho^{\beta \beta} \right)^{1/2}.
\end{equation}
In particular, if $\alpha = \beta$, we recover the Hoffman-Ostenhof inequality
\[
	\| \nabla \sqrt{\rho^{\alpha \alpha}} \|_{L^2}^2 \le \Tr( - \Delta \gamma^{\alpha \alpha}).
\]
With the Sobolev embedding $ H^1(\R^3) \hookrightarrow L^6(\R^3)$, we deduce
\[
	\| \rho^{\alpha \alpha} \|_{L^3} \le C \ \Tr(- \Delta \gamma^{\alpha \alpha} ) .
\]
Then, using the fact that $\left(\tau^\alpha + \tau^\beta \right)^{1/2} \in L^2(\R^3)$ and $\left(\rho^{\alpha \alpha} + \rho^{\beta \beta}\right)^{1/2} \in L^6(\R^3)$ and the H\"older inequality, it follows from (\ref{nabla_rho}) that
\begin{equation} \label{nabla_rho_2}
	\| \nabla \rho^{\alpha \beta} \|_{L^{3/2}} \le \| (\tau^\alpha + \tau^\beta )^{1/2} \|_{L^2} \  \| (\rho^{\alpha \alpha} + \rho^{\beta \beta})^{1/2} \|_{L^6} \le 4C \  \Tr( - \Delta \gamma) .
\end{equation}

For $\rho^{+/-}$, we use the exact expression of the eigenvalues of a $2 \times 2$ hermitian matrix:
\begin{equation} \label{rhopm}
	\rho^{+/-} = \dfrac12 \left(\rho \pm \sqrt{\rho^2 - 4 \det(R)} \right) = \dfrac12 \left( \rho \pm \sqrt{(\rho^{\spinup \spinup} - \rho^{\spindown \spindown})^2 + 4 | \rho^{\spinup \spindown} |^2} \right).
\end{equation}
Noticing that, if $f$ and $g$ are non negative,
\[
	| \nabla \sqrt{f + g} | = \dfrac{| \nabla f + \nabla g |}{2 \sqrt{ f + g}} \le \dfrac{| \nabla f | }{2 \sqrt{f + g}} + \dfrac{| \nabla g |}{2 \sqrt{f + g}} \le \dfrac{| \nabla f |}{2 \sqrt{f}} + \dfrac{| \nabla g|}{2 \sqrt{g}} = | \nabla \sqrt{f} | + | \nabla \sqrt{g} | ,
\]
we differentiate (\ref{rhopm}) to get
\begin{align*}
	 | \nabla \rho^{+/-}| 
	 	& \le \dfrac12 | \nabla \rho | + \dfrac12 \left| \nabla \sqrt{ (\rho^{\spinup \spinup} - \rho^{\spindown \spindown})^2 + 4 | \rho^{\spinup \spindown} |^2} \right| \\
		&  \le \dfrac12 | \nabla \rho^{\spinup \spinup}| + \dfrac12 | \nabla \rho^{\spindown \spindown} | + \dfrac12 \left( | \nabla \rho^{\spinup \spinup} | + | \nabla \rho^{\spindown \spindown}| + 2 \big| \nabla  | \rho^{\spinup \spindown}|  \big| \right) .
\end{align*}
All the terms on the right-hand side are in $L^{3/2} (\R^3)$ and of norm bounded by $C \Tr( - \Delta \gamma)$, hence the same holds for $\nabla \rho^{+/-}$.\\

Moreover, $\gamma$ is in $\cP_\lambda$, so that $\Tr( \gamma) = \int_{\R^3} \rho = \lambda$. We get from the inequality $2|ab| \le |a |^2 + |b|^2$ that
\begin{equation} \label{MajorationRho}
	\fl | \rho^{\alpha \beta} | =  \left| \sum_{k=1}^\infty n_k \phi_k^\alpha(\br) \overline{\phi_k^\beta(\br)} \right | \le \sum_{k=1}^\infty \dfrac{n_k}{2} \left( | \phi_k^\alpha |^2 + | \phi_k^\beta |^2 \right)  \le \sum_{k=1}^\infty n_k \left( | \phi_k^\spinup |^2 + | \phi_k^\spindown |^2 \right) = \rho.
\end{equation}
Integrating on $\R^3$ leads to $\| \rho^{\alpha \beta} \|_{L^1} \le \lambda$. From the positiveness of $R_\gamma$, it also holds that $0 \le \rho^{+/-} \le \rho$ so that $\| \rho^{+/-} \|_{L^1} \le \lambda$. We conclude from (\ref{nabla_rho_2}), the Sobolev embedding $W^{1, 3/2}(\R^3) \hookrightarrow L^3(\R^3)$, and the H\"older inequality with $1 \le p \le 3$, that
\[
	\| \rho^{\alpha \beta} \|_{L^p} \le C_p \lambda^{\frac{3-p}{2p}} \Tr (- \Delta \gamma)^{\frac{3(p-1)}{2p}}
\]
and similarly for $\rho^{+/-}$.
\end{proof}


\begin{proof}[Proof of Lemma \ref{lemma2}]
~\\
We prove that $I_\lambda > -\infty$. The proof is similar for $I_{\lambda}^\infty$. Let $\lambda > 0$, and $\gamma \in \cP_\lambda$. Under conditions (\ref{cond_g}), a straightforward calculation shows that
\begin{align*}
	\left| E_{\xc}^{\LSDA}(\rho^+, \rho^-) \right| & \le C \left( \int_{\R^3} (\rho^+)^{p^-} + \int_{\R^3} (\rho^+)^{p^+} \right) + C \left( \int_{\R^3} (\rho^-)^{p^-} + \int_{\R^3} (\rho^-)^{p^+} \right) \\
		& \le 2 C \left( \int_{\R^3} \rho^{p^+} + \int_{\R^3} \rho^{p^-} \right) ,
\end{align*}
where $p^{+/-} := 1 + \beta^{+/-} < 5/3$. We used the fact that $R_\gamma$ is a positive hermitian matrix, so that $0 \le \rho^{+/-} \le \rho$. Therefore, because $J(\rho) \ge 0$, we have the estimate:
	\[
		\cE(\gamma) \ge \frac12 \Tr( -\Delta \gamma) - C_1  \| U \|_{L^{\frac32 + \epsilon} + L^{\infty}} \| R \|_{L^{1} \cap L^{3-\epsilon'}} - C_2 \left( \| \rho \|_{L^{p^+}}^{p^+}  + \| \rho \|_{L^{p^-}}^{p^-} \right)  .
	\] 
	With Lemma \ref{lemma1}, it follows
	\[
		\cE(\gamma) \ge \frac12 \Tr( -\Delta \gamma) - \widetilde C_1 \| U \|_{L^{\frac32 + \epsilon} + L^{\infty}} \left( 1 +  \Tr(-\Delta \gamma)^{\alpha_1} \right) - C_2 \left(\Tr(-\Delta \gamma)^{\alpha_2} + \Tr(-\Delta \gamma)^{\alpha_3} \right)
	\]
	with $0 \le \alpha_1, \alpha_2, \alpha_3 < 1$.	The function $Y \mapsto Y - \widetilde C_1 (1+Y^{\alpha_1}) - C_2 Y^{\alpha_2} - C_2 Y^{\alpha_3}$ goes to $+\infty$ when $Y$ goes to $+ \infty$ for  $0 \le \alpha_1, \alpha_2, \alpha_3 < 1$. Hence, $\mathcal{\cE} \ge -C$ for all $\gamma \in \cP_\lambda$. It also follows from the above inequality that if $(\gamma_n)$ is a minimizing sequence for $I_\lambda$, then $\Tr( - \Delta \gamma_n)$ is uniformly bounded. In particular, $(\gamma_n)$ is a bounded sequence of $\cB$.

\end{proof}


\begin{proof}[Proof of Lemma \ref{lemma3}]
~\\
Let $(\gamma_n)_{n \in \N^*}$ be a bounded sequence in $\cB$. According to Lemma \ref{lemma1}, the sequences $(\rho^{\alpha \beta}_n)$ for $\alpha, \beta \in \{ \spinup, \spindown \}^2$ and $(\rho^{+/-}_n)$ are bounded in $W^{1, 3/2}(\R^3)$. In virtue of the Banach-Alaoglu theorem, up to a subsequence, the sequence $(\gamma_n)$ converges to some $\gamma_0 \in \cB$ for the weak-$*$ topology of $\cB$, and $(\rho^{\alpha \beta}_n)$ and $\rho^{+/-}_n$ converge for the weak topology of $W^{1,3/2}(\R^3)$. To identify the limits, we recall that, for any compact operator $K$ on $L^2(\R^3, \C^2)$,
\begin{equation} \label{K}
	\Tr( \gamma_n K) \xrightarrow[n \to \infty]{} \Tr(\gamma K) \quad \text{and} \quad \Tr(| \nabla | \gamma_n | \nabla | K ) \xrightarrow[n \to \infty]{}  \Tr(| \nabla | \gamma | \nabla | K ).
\end{equation}
Choose $W \in \cC^\infty_0(\R^3, \R)$. The operator $(1 + | \nabla |)^{-1} W (1 + | \nabla |)^{-1}$ is compact and in the Schatten class $\fs_p$ for $p > \frac32$ according to the Kato-Simon-Seiler inequality \cite{Simon1979}. Taking successively in (\ref{K})
\[
	K = \begin{pmatrix} W & 0 \\ 0 & 0 \end{pmatrix},  \quad K = \begin{pmatrix} 0 & 0 \\ 0 & W \end{pmatrix}, \quad K = \begin{pmatrix} 0 & W \\ W & 0 \end{pmatrix} \quad \text{and} \quad K = \begin{pmatrix} 0 & \ri W \\ -\ri W & 0 \end{pmatrix},
\]
we obtain that, for the first choice of $K$,
\begin{equation} \label{W}
\begin{aligned}
	\int_{\R^3} \rho_n^{\spinup \spinup} W & = \Tr( \gamma_n W) = \Tr \left( (1 + | \nabla |) \gamma_n (1 + | \nabla |) \cdot (1 + | \nabla |)^{-1} W (1 + | \nabla |)^{-1} \right) \\
		& \xrightarrow[n \to \infty]{} \Tr \left( (1 + | \nabla |) \gamma_0 (1 + | \nabla |) \cdot (1 + | \nabla |)^{-1} W (1 + | \nabla |)^{-1} \right) = \int_{\R^3} \rho_0^{\spinup \spinup} W
\end{aligned}
\end{equation}
and similarly for $\rho_0^{\spindown \spindown}$, $\Re (\rho_0^{\spinup \spindown})$ and $\Im( \rho_0^{\spinup \spindown})$. We deduce that $(\rho^{\alpha \beta}_n)$ converges to $\rho_0^{\alpha \beta}$ in $\cD'(\R^3, \C)$ for all $\alpha, \beta \in \{ \spinup, \spindown \}^2$. Identifying the limits, the convergences hold also weakly in $W^{1,3/2}(\R^3)$, strongly in $L^p_{\loc}(\R^3)$ for $1 \le p < 3$, and almost everywhere, in virtue of the Sobolev embedding theorem. From formula (\ref{rhopm}) and the pointwise convergence of $(\rho_n^{\alpha \beta})$ to $\rho_0^{\alpha \beta}$, we also deduce that $(\rho_n^{+/-})$ pointwise converges to $\rho_0^{+/-}$. Again, by identifying the limits, the convergence also holds weakly in $W^{1,3/2}(\R^3)$ and strongly in $L^p_{\loc}(\R^3)$ for $1 \le p < 3$. \\

Then, let $\chi \in \cC^\infty_0(\R)$ be a cut-off function such that $\chi(x) = 1$ if $| x | < 1$ and $\chi(x) = 0$ if $x \ge 2$. We take $W_A = \chi(x/A)$ in (\ref{W}), and let $A$ go to infinity to obtain that
\begin{equation} \label{CvL1}
	\rho_0^{\spinup \spinup} \in L^1(\R^3) \quad \text{and} \quad \int_{\R^3} \rho_0^{\spinup \spinup} \le \liminf_{n \to \infty} \int_{\R^3} \rho_n^{\spinup \spinup},
\end{equation}
and similarly for $\rho_0^{\spindown \spindown}$. 
Now, if $\gamma_n \in \cP_\lambda$ and $\gamma_0 \in \cP_\lambda$, we get
\[
	\lambda = \int_{\R^3} \rho_0 = \int_{\R^3} \rho_0^{\spinup \spinup} + \rho_0^{\spindown \spindown} \le  \int_{\R^3} \rho_n^{\spinup \spinup} + \rho_n^{\spindown \spindown} = \lambda,
\]
and the inequality (\ref{CvL1}) is an equality. Therefore, $(\rho_n)$ converges to $\rho_0$ strongly in $L^1(\R^3)$. We deduce from (\ref{MajorationRho}) and $0 \le \rho_n^{+/-} \le \rho_n$ that $\rho_n^{\spinup \spindown}$ and $\rho_n^{+/-}$ are bounded in $L^1(\R^3)$. A classical application of the dominated convergence theorem then leads to the fact that $\rho_n^{\alpha \beta}$ converges  to $\rho_0^{\alpha \beta}$ strongly in $L^1(\R^3)$ for $\alpha, \beta \in \{ \spinup, \spindown \}^2$, and that $\rho_n^{+/-}$ converges strongly to $\rho_0^{+/-}$ in $L^1(\R^3)$. Finally, the strong convergence still holds in $L^p(\R^3)$ for $1 \le p < 3$ according to the H\"older inequality. \\
\indent The proof for the energy is similar to the one in \cite[Lemma 3]{Anantharaman2009}. We do not repeat it here, but notice that the strong convergence of $(\rho_n^{+/-})$  to $\rho_0^{+/-}$ in $L^p(\R^3)$ for $1 \le p < 3$ is needed for the convergence of the exchange-correlation functional.

\end{proof}


\begin{proof} [Proof of Lemma \ref{lemma5}]
~\\
We first prove that there exists $\lambda_0$ small enough such that for all $0 < \lambda \le \lambda_0$, $I_\lambda^\infty < 0$. We use a scaling argument. Let $\phi \in \cC^\infty_0(\R^3, \C)$ be such that $\| \phi \|_{L^2} = 1$, and let $\phi_\sigma = \sigma^{3/2} \phi(\sigma \cdot)$ for $\sigma > 0$. Note that $\| \phi_\sigma \|_{L^2} = 1$. For $\lambda \le 1$, we introduce
\[
	\gamma_{\lambda \sigma}(\br, \br') = \lambda \begin{pmatrix} \phi_\sigma(\br) \overline{\phi_\sigma} (\br') & 0 \\ 0 & 0 \end{pmatrix}
\]
so that $\gamma_{\lambda \sigma} \in \cP_{\lambda}$ for all $0 < \lambda \le 1$ and $\sigma > 0$. Using (\ref{cond_g}), there exists $1 \le \alpha< 3/2$ such that $E_{\xc}^\LSDA(| \phi_{\lambda \sigma} |^2, 0) \le - C \lambda^\alpha \sigma^{3 (\alpha-1)}$. Direct calculations lead to
\begin{align*}
	\cE^\infty(\gamma_{\lambda \sigma}) & = \dfrac{ \lambda \sigma^2}{2} \int_{\R^3} | \nabla \phi |^2 + \lambda^2 \sigma J( | \phi |^2) + \int_{\R^3} E_{\xc}^\LSDA(| \phi_{\lambda \sigma} |^2, 0) \\
	 & \le \dfrac{\lambda \sigma^2}{2} \int_{\R^3} | \nabla \phi |^2 + \lambda^2 \sigma J( | \phi |^2) - C \lambda^\alpha \sigma^{3(\alpha-1)}.
\end{align*}
It is easy to check that under the condition $\alpha < 3/2$, there exists $\lambda_0 > 0$ such that for all $0 < \lambda \le \lambda_0$, there exists $\sigma$ such that $\cE(\gamma_{\lambda \sigma}) < 0$. In particular, $I_{\lambda}^\infty \le \cE^\infty(\gamma_{\lambda \sigma}) < 0$. \\


\textit{(ii)} We now prove that $I_\lambda < I_\lambda^\infty$, for all $\lambda > 0$. Let $(\gamma_n)$ be a minimizing sequence for $I_\lambda^\infty$.\\
We first suppose that
\[
	\forall A > 0, \quad \lim_{n \to \infty} \sup_{x \in \R^3} \int_{x + B_A} \rho_n = 0,
\]
where $B_A$ is the ball of radius $A$ centered at the origin. Because $(\rho_n)$ is bounded in $W^{1, 3/2}$ according to Lemma \ref{lemma2}, we deduce from \cite[Lemma I.1]{Lions1984} that $(\rho_n)$ converges to $0$ strongly in $L^p(\R^3)$ for $1 < p < 3$. Also, because of (\ref{MajorationRho}), the components of $R_n$ and its eigenvalues converge to $0$ strongly in $L^p(\R^3)$ for $1 < p < 3$. Similarly to \cite{Anantharaman2009}, we deduce that
\[
	\hspace{-2em} I_\lambda^\infty = \liminf_{n \to \infty} \cE^\infty(\gamma_n) = \liminf_{n \to \infty} \left\{ \dfrac12 \Tr( - \Delta \gamma_n) + J( \rho_n) + E_{\xc}^\LSDA(\rho^+_n, \rho^-_n) \right\} = \liminf_{n \to \infty} \dfrac12 \Tr( - \Delta \gamma_n) \ge 0
\]
which contradicts the first point. Therefore
\begin{equation} \label{eta}
	\exists A, \eta > 0, \quad \forall n \in \N, \quad \exists x_n \in \R^3, \quad \int_{x_n + B_A} \rho_n \ge \eta.
\end{equation}
Up to translations of the $\gamma_n$'s, we can assume without loss of generality that $x_n = 0$.

We now introduce $\widetilde \gamma_n$, the flipped version of $\gamma_n$ introduced in (\ref{transformation}). Using (\ref{ControlZeeman}) and the fact that $V(\br) \le - \dfrac{z_1}{\br}$, we get
\begin{align*}
	\cE(\gamma_n) + \cE(\widetilde \gamma_n) & = \Tr( - \Delta \gamma_n) + 2 J (\rho_n) + 2 E_{\xc}^\LSDA(\rho^+_n, \rho^-_n) + 2 \int_{\R^3} V \rho_n \\
		& = 2 \cE^\infty(\gamma_n) + 2 \int_{\R^3} V \rho_n 
		\le 2\cE^\infty(\gamma_n) - 2 \int_{B_R} \dfrac{z_1}{| \br |} \rho_n \le 2\cE^\infty(\gamma_n) - 2 \dfrac{z_1}{R} \eta .
\end{align*}
Hence, either $\cE(\gamma_n)$ or $\cE(\widetilde \gamma_n)$ is smaller than $\cE^\infty(\gamma_n ) - z_1 R^{-1} \eta$. Therefore, $I_\lambda \le I_\lambda^\infty - z_1 R^{-1} \eta < I_\lambda^\infty$. \\


\textit{(iii)} Let us prove that for $0 < \mu < \lambda$, it holds that $I_\lambda \le I_\mu + I_{\lambda - \mu}^\infty$. Let $\varepsilon > 0$, $\gamma \in \cP_\mu$ and $\gamma' \in \cP_{\lambda - \mu}$ be such that $I_\mu \le \cE(\gamma) \le I_\mu + \varepsilon$ and $I^\infty_{\lambda - \mu} \le \cE^\infty(\gamma') \le I^\infty_{\lambda - \mu} + \varepsilon$. By density of finite-rank one-body density matrices in $\cB$, and density of $\cC^\infty_0 (\R^3, \C^2)$ in $H^1(\R^3, \C^2)$, we can assume that $\gamma$ and $\gamma'$ are both of the form
\[
	\gamma^{(')} = \sum_{i=1}^M n_k^{(')} | \Phi_k^{(')} \ket \ \bra \Phi_k^{(')} | \quad \text{with} \quad \Phi_k^{(')} \in \cC^\infty_0(\R^3, \C^2).
\]
We consider $\gamma_n := \gamma + \tau_{n\be} \gamma ' \tau_{-n \be} \in \cP_\lambda$ and $\gamma_n^\sharp := \gamma + \tau_{n\be} \widetilde \gamma' \tau_{-n \be} \in \cP_\lambda$ where $\tau_{\bx} f (\br) = f( \br - \bx)$, and $\be$ is a non-null vector. We recall that $\widetilde \gamma'$ is the flipped transformation of $\gamma'$, as introduced in (\ref{transformation}). For $n_0$ large enough, and for $n \ge n_0$, the supports of the $\Phi_k$'s and of the $\tau_{n\be} \Phi_k'$'s are disjoint, so that $\gamma_n$ and $\gamma_n^\sharp$ are in $\cP_\lambda$ for all $n \ge n_0$. Also, for $n$ large enough, $J(\rho_n) \le J(\rho) + J(\rho') + \varepsilon$. Altogether, we get, for $n$ large enough,
\begin{align*}
	\cE(\gamma_n) + \cE( \gamma_n^\sharp) & = 2 \cE(\gamma) + 2 \cE^\infty(\gamma') + 2 \int V \rho'(\cdot - n \be) + 2  \varepsilon \le 2 \cE(\gamma) + 2 \cE^\infty(\gamma') + 2 \varepsilon \\
		& \le 2 I_\mu + 2 I_{\lambda - \mu}^\infty + 6 \varepsilon .
\end{align*}
As before, either $\cE(\gamma_n)$ or $\cE(\gamma_n^\sharp)$ is smaller than  $I_\mu + I_{\lambda - \mu} + 3 \varepsilon$, hence $I_\lambda \le I_\mu + I_{\lambda - \mu}$. \\


 \textit{(iv)} and \textit{(i)} The fact that $\lambda \mapsto I_\lambda$ and $\lambda \mapsto I_\lambda^\infty$ are non increasing, and that $I_{\lambda}^\infty < 0$ and $I_\lambda < 0$ for all $\lambda > 0$ can be read from the other statements.\\

\end{proof}


\begin{proof} [Proof of Lemma \ref{lemma6}]
~\\
	Let $\lambda > 0$, and let $(\gamma_n)_{n \in \N^*} \in \cP_\lambda$ be a minimizing sequence for $I_\lambda$. According to Lemma \ref{lemma2}, up to a subsequence, we can assume that $(\gamma_n)$ converges to some $\gamma_0 \in \cB$ for the weak-$*$ topology of $\cB$. \\
	
	\textit{(i)} The fact that $\alpha \le \lambda$ can be directly deduced from (\ref{CvL1}). \\
	
	\textit{(ii)} Suppose that $\alpha = 0$, so that $\gamma = 0$. Then, we have $I_\lambda = \liminf \cE( \gamma_n) = \cE(\gamma_0) = 0$ (we used the continuity of $\cE$, which can be proved similarly to \cite{Anantharaman2009}). This contradicts the first point of Lemma \ref{lemma4}. Hence, $\alpha \neq 0$. \\
	
	\textit{(iii)} Suppose that $0 < \alpha < \lambda$. Following \cite{Anantharaman2009,Frank2007}, we let $\chi, \xi \in \cC^\infty_0(\R^3, \R^+)$ be radial functions such that $\chi^2 + \xi^2 = 1$, with $\chi(0) = 1$, $\chi < 1$ on $\R^3 \setminus{\{0\}}$, $\chi(x) = 0$ for $| x | > 1$, $\| \nabla \chi \|_{L^\infty} \le 2$ and $\| \nabla \xi \|_{L^\infty} \le 2$. We introduce $\chi_A (x) := \chi( x /A ) $ and $\xi_A(x) := \xi(x / A)$ and finally $\gamma_{n,A} := \chi_A \gamma_n \chi_A$. With those notations, $A \mapsto \Tr (\gamma_{n, A})$ is a continuous and increasing function from $0$ to $\lambda$. Therefore, there exists $A_n$ such that $\gamma_{n, A_n}$ is in $\cP_\alpha$. \\
 \indent The sequence $(A_n)$ goes to infinity. Otherwise, we would have for $A$ large enough and according to~(\ref{CvL1}),
\[
	\int_{\R^3} \rho_0 \chi^2_A = \lim_{n \to \infty} \int_{\R^3} \rho_n \chi^2_A \ \ge \lim_{n \to \infty}  \int_{\R^3} \rho_n \chi^2_{A_n} \  = \alpha = \int_{\R^3} \rho_0
\]
which is impossible, for $| \chi^2_A| < 1$ on $\R^3$. \\

We introduce $\gamma_{1,n} := \chi_{A_n} \gamma_n \chi_{A_n}$ and $\gamma_{2,n} := \xi_{A_n} \gamma_{ n} \xi_{A_n}$. Note that $\gamma_{1,n} \in \cP_{\alpha}$ and $\gamma_{2,n} \in \cP_{\lambda - \alpha}$, and that $\rho_n = \rho_{1,n} + \rho_{2,n}$. Also, direct calculations lead to
\begin{equation} \label{KE}
	\Tr(-\Delta \gamma_{1,n}) + \Tr( - \Delta \gamma_{2,n}) \le \Tr ( - \Delta \gamma_n) + 8 \dfrac{\lambda}{A_n^2}.
\end{equation}
Hence, $(\gamma_{1,n})$ and $(\gamma_{2,n})$ are bounded in $\cB$. According to Lemma \ref{lemma3}, up to a subsequence, $(\gamma_{1,n})$ converges for the weak-$*$ topology of $\cB$. In this case, for $\Phi = (\phi^\spinup, \phi^\spindown) \in \cC^\infty_0(\R^3, \C^2)$, it holds that
\[
	\Tr ( \gamma_{1,n} | \Phi \ket | \bra \Phi | ) = \int_{\R^3} \rho_{1,n}^{\spinup \spinup} | \phi^\spinup |^2 + \int_{\R^3} \rho_{1,n}^{\spindown \spindown} | \phi^\spindown |^2 = \int_{\R^3} \chi_{A_n}^2 \rho_n^{\spinup \spinup} | \phi^\spinup |^2 +  \int_{\R^3} \chi_{A_n}^2 \rho_n^{\spindown \spindown} | \phi^\spindown |^2 .
\]
For $n$ large enough, the support of $\Phi$ is inside the support of $\chi_{A_n}$, and
\[
	\Tr ( \gamma_{1,n} | \Phi \ket  \bra \Phi | )  = \Tr( \gamma_n | \chi_{A_n} \Phi \ket \bra \Phi \chi_{A_n}|)    \xrightarrow[n \to \infty]{} \Tr ( \gamma | \Phi \ket | \bra \Phi |).
\]
We deduce that $(\gamma_{1,n})$ converges to $\gamma_0$ for the weak-$*$ topology of $\cB$. Finally, because $\gamma_{1,n} \in \cP_\alpha$ and $\gamma_0 \in \cP_\alpha$, $\rho_{1,n}$ converges strongly to $\rho_0$ in $L^p(\R^3)$ for $1 \le p < 3$, and $\cE(\gamma_0) \le \liminf \cE(\gamma_{1,n})$ according to Lemma \ref{lemma3}.\\

Let us look more closely to $\gamma_{2,n}$. 
Because $(\rho_{1,n})$ converges to $\rho_0$ strongly in $L^p(\R^3)$ and $(\rho_n)$ converges to $\rho_0$ strongly in $L^p_{\loc}(\R^3)$ for $1 \le p < 3$, we obtain that $\rho_{2,n} = \rho_n - \rho_{1,n}$  (and thus all the components of $R_{2,n}$ and its eigenvalues) converges strongly to $0$ in $L^p_{\loc}(\R^3)$ for $1 \le p < 3$. Also, it holds that $\rho_{1,n}^{+/-} + \rho_{2,n}^{+/-} = \rho_n^{+/-} $. Using (\ref{KE}) and the fact that $\iint \rho_{1,n}(\br) \rho_{2,n}(\br') | \br- \br' |^{-1} \rd^3 \br \rd^3 \br' \ge 0$, we obtain
\begin{align*}
	\cE(\gamma_n) & = \dfrac12 \Tr( - \Delta \gamma_n) + J(\rho_n) + \int_{\R^3} \tr_{\C^2} \left[ U R_n \right] + E_{\xc}^\LSDA(\rho_n^+, \rho_n^-) \\
	& \ge \dfrac12 \Tr ( - \Delta \gamma_{1,n}) + \dfrac12 \Tr ( - \Delta \gamma_{2,n}) - 4 \dfrac{\lambda}{A_n^2} + J(\rho_{1,n}) + J(\rho_{2,n}) + \\
	& + \int_{\R^3}\tr_{\C^2} \left[ U R_{1,n} \right] + \int_{\R^3} \tr_{\C^2} \left[ U R_{2,n} \right] 
	 + E_{\xc}^\LSDA (\rho_{1,n}^+ + \rho_{2,n}^+, \rho_{1,n}^- + \rho_{2,n}^-) \\
	& \ge \cE(\gamma_{1,n}) + \cE^\infty(\gamma_{2,n}) - 4 \dfrac{\lambda}{A_n^2} + \int_{\R^3} \tr_{\C^2} \left[ U R_{2,n} \right] + \\
	& \quad + E_{\xc}^\LSDA (\rho_{1,n}^+ + \rho_{2,n}^+, \rho_{1,n}^- + \rho_{2,n}^-) - E_{\xc}^\LSDA(\rho_{1,n}^+, \rho_{1,n}^-) - E_{\xc}^\LSDA(\rho_{2,n}^+, \rho_{2,n}^-).
\end{align*}
We first consider the term $\int \tr_{\C^2} \left[U R_{2,n} \right]$. We have for $A \ge 0$, (we use, for a matrix $M$, the notation $| M |$ for the sum of the absolute values of the entries of $M$)
\begin{align*}
	\left| \int_{\R^3} \tr_{\C^2} \left[U R_{2,n} \right] \right| & = \left| \int_{B_A} \tr_{\C^2} \left[ U R_{2,n} \right] \right| + \left| \int_{(B_A)^c} \tr_{\C^2} \left[ U R_{2,n} \right] \right| \\
		& \le \| U \|_{L^{3/2 + \epsilon} + L^\infty(B_A)} \| R_{2,n} \|_{L^1 \cap L^{3-\epsilon'} (B_A)} + \sup_{x \in (B_A)^c} |U(x)| \int_{(B_A)^c} | R_{2,n}|  \\
		& \le  \| U \|_{L^{3/2 + \epsilon} + L^\infty(\R^3)} \| R_{2,n} \|_{L^1 \cap L^{3-\epsilon'} (B_A)} + \sup_{x \in (B_A)^c} |U(x)| \int_{\R^3} | R_{2,n} |.
\end{align*}
Using inequality (\ref{MajorationRho}), and the fact that $\int \rho_{2,n}^{\alpha \beta} \le \lambda$, we get an inequality of the form
\[
	\left| \int_{\R^3} \tr_{\C^2} \left[U R_{2,n} \right] \right| \le C_1 \| R_{2,n} \|_{L^1 \cap L^{3-\epsilon'} (B_A)} + C_2  \sup_{x \in (B_A)^c} |U(x)|
\]
with $C_1$ and $C_2$ independent of $A$ and $n$. Because all entries of $U$ are vanishing at infinity, we can first choose $A$ large enough to control the second term, and then use the convergence of $R_{2,n}$ to $0$ strongly in $L^p(B_A)$ for $1 \le p < 3$, to establish the convergence of the right-hand-side to $0$. \\
For the last term, using (\ref{cond_g}), it holds (we write $g_2(\rho) = g(2\rho)$)

\begin{equation} \label{boundEXC}
\begin{aligned}
	& E_{\xc}^\LSDA (\rho_{1,n}^+ + \rho_{2,n}^+, \rho_{1,n}^- + \rho_{2,n}^-) - E_{\xc}^\LSDA(\rho_{1,n}^+, \rho_{1,n}^-) - E_{\xc}^\LSDA(\rho_{2,n}^+, \rho_{2,n}^-) = \\
	& \quad \dfrac12 \left[ \int_{\R^3} \left( g_2(\rho_{1,n}^+ + \rho_{2,n}^+) - g_2(\rho_{1,n}^+) - g_2( \rho_{2,n}^+) \right) + \int_{\R^3} g_2(\rho_{1,n}^- + \rho_{2,n}^-) - g_2(\rho_{1,n}^-) - g_2(\rho_{2,n}^-)      \right]	.
\end{aligned}
\end{equation}
	Then, we get (dropping the super-script $+/-$ for the sake of clarity)
\begin{align*}
	& \left| \int_{\R^3}g_2(\rho_{1,n} + \rho_{2,n}) - g_2(\rho_{1,n}) - g_2( \rho_{2,n}) \right|   \\
	& \quad \le  \int_{B_A} \left|  g_2(\rho_{1,n} + \rho_{2,n}) - g_2(\rho_{1,n}) \right| + \int_{B_A} | g_2 (\rho_{2,n}) |  + \\
	& \quad + \int_{(B_A)^c} \left|  g_2(\rho_{1,n} + \rho_{2,n}) - g_2(\rho_{2,n}) \right| + \int_{(B_A)^c} | g_2 (\rho_{2,n}) | \\
	& \quad \le C \left( \int_{B_A} \rho_{2,n} \left( \rho_n^{p^+} + \rho_n^{p^-} \right)
	 +  \int_{B_A} \left( (\rho_{2,n})^{p^-} + (\rho_{2,n})^{p^+}\right) \right)\\
	& \quad  + C \left( \int_{(B_A)^c} \rho_{1,n} \left( \rho_n^{p^+} + \rho_n^{p^-} \right) 
	 +  \int_{(B_A)^c} \left( (\rho_{1,n})^{p^-} + (\rho_{1,n})^{p^+} \right) \right).
\end{align*}
We recall that $p^{+/-} = 1 + \beta^{+/-} < 5/3$. Because $(\rho_{1,n})$ and $(\rho_n)$ are bounded in $L^p(\R^3)$ for $1 \le p < 3$, and because $(\rho_{2,n})$ converges to $0$ in $L^p_{\loc}(\R^3)$ for $1 \le p < 3$, we deduce that (\ref{boundEXC}) goes to $0$ when $n$ goes to infinity (first take $A$ large enough, then $n$ large enough, as before).\\
Altogether, for $\epsilon > 0$, for $n$ large enough,
\[
	\cE(\gamma_n) \ge \cE(\gamma_{1,n}) + \cE^\infty(\gamma_{2,n}) - 3 \epsilon \ge I_\alpha + I^\infty_{\lambda - \alpha} - 3 \epsilon .
\]
Therefore, $\cE(\gamma_n) \ge I_\alpha + I^\infty_{\lambda - \alpha}$, and $I_\lambda \ge I_\alpha + I^\infty_{\lambda - \alpha}$. The third point of Lemma \ref{lemma4} states that $I_\lambda \le I_\alpha + I_{\lambda - \alpha}^\infty$. Hence $I_\lambda = I_\alpha + I^\infty_{\lambda - \alpha}$, and $(\gamma_{2,n})$ is a minimizing sequence for $I_{\lambda - \alpha}^\infty$.

As in the proof of Lemma \ref{lemma4}, we have (\ref{eta}):
\[
	\exists A, \eta > 0, \quad \forall n \in \N, \quad \exists x_n \in \R^3, \quad \int_{x_n + B_A} \rho_{2,n} \ge \eta.
\]
We let $\gamma_{2,n}' = \tau_{x_n} \gamma_{2,n} \tau_{-x_n}$. Then, $(\gamma_{2,n})$ is bounded for the weak-$*$ topology of $\cB$, and converges, up to a subsequence, to some $\gamma_0'$ satisfying $\Tr(\gamma_0') \ge \eta$. Let $\beta := \Tr( \gamma_0')$. We can repeat the same arguments as before and truncate $\gamma_{2,n}'$ to ensure that $\Tr( \chi_{A_n} \gamma_{2,n} \chi_{A_n}) = \beta$. We deduce as before that $\gamma_0'$ is a minimizer for $I_{\beta}^\infty$, and that $I_\lambda = I_\alpha + I_\beta^\infty + I_{\lambda - \alpha - \beta}^\infty$.
\end{proof}


\begin{proof} [Proof of Lemma \ref{lemma7}]
~\\
Let us first derive the expression of $H_{\gamma_0}$. Suppose that $\gamma_0 \in \cP_\lambda$ is a minimizer for $I_\lambda$. Then for $\gamma \in \cP_\lambda$ and $0 \le t \le 1$, it holds $\cE(t \gamma + (1 - t) \gamma_0) \ge \cE(\gamma_0)$. In particular, one must have
\begin{equation}  \label{condGamma0}
	\dfrac{\partial \cE(t \gamma + (1 - t) \gamma_0)}{\partial t} \Big|_{t=0} \ge 0 .
\end{equation}
To perform the calculations, we use the explicit formula (\ref{rhopm}) for $\rho^{+/-}$, and get
\begin{align*}
	& \dfrac{\partial \left( t \rho  + (1 - t) \rho_0 \right)^{+/-}}{\partial t} \Big|_{t=0} = \\
	& \qquad \dfrac12 \tr_{\C^2} \left( \left[ \begin{pmatrix} 1 & 0 \\ 0 & 1 \end{pmatrix} 
	\pm  
	\dfrac{1}{ \sqrt{(\rho_0^{\spinup \spinup} - \rho_0^{\spindown \spindown})^2 + 4 | \rho_0^{\spinup \spindown} |^2}} \begin{pmatrix} \rho_0^{\spinup \spinup} - \rho_0^{\spindown  \spindown} & 2 \rho_0^{\spinup \spindown} \\ 2 \rho_0^{\spindown \spinup} & \rho_0^{\spindown \spindown} - \rho_0^{\spinup \spinup} \end{pmatrix} \right] 
	(R - R_0) \right) .
\end{align*}
Similarly to \cite{Anantharaman2009,Cances2008}, we conclude that 
\[
	\dfrac{\partial \cE(t \gamma + (1 - t) \gamma_0)}{\partial t} \Big|_{t=0} = \Tr \left( H_{\gamma_0} (\gamma - \gamma_0) \right)
\]
with
\begin{equation}  \label{Hgamma}
\begin{aligned}
	H_{\gamma_0} = & \left( -\frac12 \Delta + \rho_0 \star | \cdot |^{-1} \right) \b1_2 + U \\
		& + \dfrac{g'(\rho_0^+)}{2}   \left[ \begin{pmatrix} 1 & 0 \\ 0 & 1 \end{pmatrix} 
	+  
	\dfrac{1}{ \sqrt{(\rho_0^{\spinup \spinup} - \rho_0^{\spindown \spindown})^2 + 4 | \rho_0^{\spinup \spindown} |^2}} \begin{pmatrix} \rho_0^{\spinup \spinup} - \rho_0^{\spindown  \spindown} & 2 \rho_0^{\spinup \spindown} \\ 2 \rho_0^{\spindown \spinup} & \rho_0^{\spindown \spindown} - \rho_0^{\spinup \spinup} \end{pmatrix} \right]   \\
		& + \dfrac{g'(\rho_0^-)}{2}  \left[ \begin{pmatrix} 1 & 0 \\ 0 & 1 \end{pmatrix} 
	-
	\dfrac{1}{ \sqrt{(\rho_0^{\spinup \spinup} - \rho_0^{\spindown \spindown})^2 + 4 | \rho_0^{\spinup \spindown} |^2}} \begin{pmatrix} \rho_0^{\spinup \spinup} - \rho_0^{\spindown  \spindown} & 2 \rho_0^{\spinup \spindown} \\ 2 \rho_0^{\spindown \spinup} & \rho_0^{\spindown \spindown} - \rho_0^{\spinup \spinup} \end{pmatrix} \right] .
\end{aligned}
\end{equation}
Using (\ref{condGamma0}), we deduce that $\gamma_0 \in \textrm{arginf} \{ \Tr (H_{\gamma_0} \gamma), \gamma \in \cP_\lambda\}$. Finally,
\[
	\gamma_0 = \b1 \left( H_{\gamma_0} < \epsilon_F \right) + \delta \quad \text{with} \quad \delta \subset \textrm{Ker}(H_{\gamma_0} - \epsilon_F),
\]
where $\epsilon_F$ is the Fermi energy, determined by the condition $\Tr (\gamma_0) = \lambda$. \\

Let us first calculate the essential spectrum of $H_{\gamma_0}$. We recall that $H_0 = - \dfrac12 \Delta \b1_2$ has domain $H^2(\R^3, \C^2)$ and that if $u \in H^2(\R^3, \C)$, then $u$ vanishes at infinity. We also recall that for all $V \in L^{3/2}(\R^3, \C^2) + L^\infty_\epsilon( \R^3, \C^2)$, the set of functions $V$ that can be written $V = V_{3/2} + V_\infty$ with $V_{3/2} \in L^{3/2}(\R^3, \C^2)$,  $V_\infty \in L^\infty(\R^3)$ and $\|V_\infty \|_{L^\infty}$ arbitrary small, $V$ is a compact perturbation of $H_0$. In our case, we can easily check that $\widehat{\rho_0 \star | \cdot |^{-1}} = \widehat{\rho_0} | \cdot |^{-2} \in L^1(\R^3)$, so that $\rho_0 \star | \cdot |^{-1}$ vanishes at infinity. Altogether,
\begin{align*}
	\bullet & \quad \rho_0 \star | \cdot |^{-1} \in L^{3/2}(\R^3) + L^\infty_\epsilon( \R^3) \\
	\bullet & \quad U \in L^{3/2}(\R^3, \C^2) + L^\infty( \R^3, \C^2) \quad \text{and  all entries of $U$ vanishes at infinity}\\
	\bullet & \quad | g'(\rho_0^{+/-}) | \le C (\rho_0^{\beta^-} + \rho_0^{\beta^+}) \quad \text{hence} \quad g'(\rho_0^{+/-}) \in L^{3/2}(\R^3, \C^2).
\end{align*}
Therefore, according to the Weyl's theorem, the domain of $H_{\gamma_0}$ is $H^2(\R^3, \C^2)$, and $\sigma_{\textrm{ess}}(H_{\gamma_0}) = \sigma_{\textrm{ess}}(H_0) = [0, + \infty[$.

Let us now prove that $H_{\gamma_0}$ has infinitely many negative eigenvalues whenever $\lambda < Z$. First notice that the matrix
\[
	\dfrac{1}{ \sqrt{(\rho_0^{\spinup \spinup} - \rho_0^{\spindown \spindown})^2 + 4 | \rho_0^{\spinup \spindown} |^2}} \begin{pmatrix} \rho_0^{\spinup \spinup} - \rho_0^{\spindown  \spindown} & 2 \rho_0^{\spinup \spindown} \\ 2 \rho_0^{\spindown \spinup} & \rho_0^{\spindown \spindown} - \rho_0^{\spinup \spinup} \end{pmatrix}
\]
has two eigenvalues, respectively $-1$ and $1$, so that the matrices appearing into the two pairs of brackets in (\ref{Hgamma}) have $0$ and $2$ as eigenvalues, and therefore are hermitian positive. Also, recall that under the conditions (\ref{cond_g}) on $g$, it holds $g' \le 0$. Altogether, for $\psi \in \cC^\infty_0 (\R^3, \C)$, $\Psi = (\psi, \psi)^T \in \cC^\infty_0(\R^3, \C^2)$, and $\widetilde \Psi$ defined as in (\ref{transformation}), it holds that
\begin{align*}
	\bra \Psi | H_{\gamma_0} | \Psi \ket + \bra \widetilde \Psi | H_{\gamma_0} | \widetilde \Psi \ket 
		& \le \left\bra \Psi \big| \left(  - \left( \frac12 \Delta + \rho_0 \star | \cdot |^{-1} \right) \b1_2 + U \right) \big| \Psi \right\ket \\
		&  \quad +  \left\bra \widetilde \Psi \big| \left(  - \left( \frac12 \Delta + \rho_0 \star | \cdot |^{-1} \right) \b1_2 + U \right) \big| \widetilde \Psi \right\ket \\
	& \le 4 \left\bra \psi \big| - \frac12 \Delta + \rho_0 \star | \cdot |^{-1} + V \big| \psi \right\ket = \bra \psi | H_1 | \psi \ket_1
\end{align*}
where $H_1 := - \frac12 \Delta + \rho_0 \star | \cdot |^{-1} + V$ acts on $L^2(\R^3, \C)$, and $V$ is defined in (\ref{V}). We used the subscript $1$ to emphasize that $\bra \cdot | \cdot \ket_1$ is the scalar product on $L^2(\R^3, \C)$, whereas $\bra \cdot | \cdot \ket$ is the one on $L^2(\R^3, \C^2)$. In virtue of \cite[Lemma 2.1]{Lions1987}, the operator $H_1$ has infinitely many negative eigenvalues of finite multiplicity whenever $\lambda < Z$. So has $H_{\gamma_0}$ by the min-max principle. Eventually, $\epsilon_F < 0$, and
\begin{equation} \label{gammaProj}
	\gamma_0 = \sum_{i=1}^{N_1} | \Phi_i \ket \ \bra \Phi_i | + \sum_{i=N_1+1}^{N_2} n_i | \Phi_i \ket \ \bra \Phi_i | \quad \text{with} \quad \bra \Phi_i | \Phi_j \ket = \delta_{ij}  \quad \text{and}  \quad H_{\gamma} \Phi_i = \epsilon_i \Phi_i .
\end{equation}
It holds $\epsilon_i < \epsilon_F$ if $i \le N_1$, and $\epsilon_i = \epsilon_F$ if $N_1 + 1 \le i \le N_2$. In the following, we set $n_i := 1$ for $ i \le N_1$.\\

Finally, we prove that all eigenvectors associated with negative eigenvalues are exponentially decreasing. 
 Any function $u$ satisfying $H_{\gamma_0} u = \lambda u$ is in $H^2(\R^3, \C^2)$, and each component of $u$ vanishes at infinity. As a byproduct, we obtain that $\rho = \sum_{i=1}^{N_2} n_i | \Phi_i |^2$ also vanishes at infinity. Finally, all the components of
\[
	\hspace{-3em} U_{\gamma_0} := \rho_0 \star | \cdot |^{-1} \b1_2 + U + \sum_{\delta = +/-} \dfrac{g'(\rho_0^\delta)}{2}   \left[ \begin{pmatrix} 1 & 0 \\ 0 & 1 \end{pmatrix} 
	+  (-1)^\delta
	\dfrac{1}{ \sqrt{(\rho_0^{\spinup \spinup} - \rho_0^{\spindown \spindown})^2 + 4 | \rho_0^{\spinup \spindown} |^2}} \begin{pmatrix} \rho_0^{\spinup \spinup} - \rho_0^{\spindown  \spindown} & 2 \rho_0^{\spinup \spindown} \\ 2 \rho_0^{\spindown \spinup} & \rho_0^{\spindown \spindown} - \rho_0^{\spinup \spinup} \end{pmatrix} \right]
\]
vanish at infinity. Recall that $H_{\gamma_0} \Phi_i = - \frac12 \Delta \Phi_i + U_\gamma \Phi_i =  \epsilon_i \Phi_i$. Multiplying this equation by $\Phi_i$ and adding all the terms with  prefactors $n_i$, it holds that
\begin{equation} \label{eqtRho1}
	\sum_{i=1}^{N_2} n_i \Phi_i^T \left( -\dfrac12 \Delta \right) \Phi_i + \sum_{i=1}^{N_2} n_i \Phi_i^T U_\gamma \Phi_i = \sum_{i=1}^{N_2} \epsilon_i n_i | \Phi_i |^2 .
\end{equation}
From the relation $\rho_0 = \sum_{i=1}^{N_2} n_i | \Phi_i |^2$, we get
\[
	\Delta \rho_0 = \sum_{i=1}^{N_2} 2 n_i \left( \Phi_i^T (\Delta \Phi_i) + | \nabla \Phi_i |^2 \right)
\]
and (\ref{eqtRho1}) becomes
\[
	- \dfrac{\Delta}{4} \rho_0 + \underbrace{\sum_{i=1}^{N_2} \dfrac{n_i}{2} | \nabla \Phi_i |^2}_{\ge 0}  + \sum_{i=1}^{N_2} n_i \Phi_i^T U_\gamma \Phi_i + \underbrace{\sum_{i=1}^{N_2} (\epsilon_F - \epsilon_i) n_i | \Phi_i |^2}_{\ge 0} - \epsilon_F \rho = 0 .
\]
Let $A$ be large enough such that, for all $\br \in \R^3$ with $| \br | \ge A$, the eigenvalues of the matrix $U_\gamma(\br)$ are between $\dfrac{\epsilon_F}{2\lambda }$ and $-\dfrac{\epsilon_F}{2 \lambda}$ (recall that $\epsilon_F < 0$). In particular, for $| \br | \ge A$, $| \Phi_i^T(\br) U_\gamma (\br) \Phi_i(\br) | \le -\dfrac{\epsilon_F}{2\lambda} | \Phi_i |^2$, and, on $(B_A)^c$,
\[
	- \dfrac{\Delta}{4} \rho + \dfrac{\epsilon_F \lambda}{2\lambda} \rho - \epsilon_F \rho \le 0 \quad \text{or} \quad -\dfrac{\Delta}{2} \rho - \epsilon_F \rho_0 \le 0.
\]
We easily deduce that $\rho_0$ decreases exponentially. Hence, the same holds true for all the $\Phi_i$'s with $1 \le i \le N_2$. A similar proof can be used for the remaining negative eigenvalues.

\end{proof}


\begin{proof} [Proof of Lemma \ref{lemma8}]
~\\
Let $\gamma_0 \in \cP_\alpha$ be a minimizer for $I_\alpha$, and $\gamma_0' \in \cP_\beta$ be a minimizer for $I_\beta^\infty$. According to Lemma \ref{lemma6}, because $\alpha < \lambda$, $\gamma_0$ has the form
\[
	\gamma_0 = \sum_{i=1}^{N_2} n_i | \Phi_i \ket \ \bra \Phi_i | \quad \text{with} \quad H_{\gamma_0} \Phi_i = \epsilon_i \Phi_i \quad \text{and} \quad \epsilon_i \le \epsilon_F < 0 .
\]
We can derive a similar expression for $\gamma_0'$, as in the proof of Lemma \ref{lemma6}:
\begin{equation} \label{gamma'}
	\gamma_0' = \sum_{i=1}^\infty n_i' | \Phi_i' \ket \ \bra \Phi_i' | \quad \text{with} \quad H^\infty_{\gamma_0'} \Phi_i' = \epsilon_i \Phi_i' \quad \text{and} \quad \epsilon_i' \le \epsilon_F' \le 0,
\end{equation}
where $H^\infty_{\gamma_0'}$ has a similar expression as $H_{\gamma_0'}$ in (\ref{Hgamma}), without the $U$ term. Note that in (\ref{gamma'}), we do not know whether $\epsilon_F' < 0$ or $\epsilon_F' = 0$. \\

First assume that $\epsilon_F' < 0$, so that $\Phi_i$ and $\Phi_i'$ are exponentially decreasing, and the sum in (\ref{gamma'}) is finite. We introduce $ \gamma_n := \min \{ 1, \| \gamma_0 + \tau_n \gamma_0' \tau_{-n} \|^{-1} \} \left( \gamma_0 + \tau_n \gamma_0' \tau_{-n} \right)$ and $ \gamma_n^\sharp := \min \{ 1, \| \gamma_0 + \tau_n \widetilde \gamma_0' \tau_{-n} \|^{-1} \} \left( \gamma_0 + \tau_n \widetilde \gamma_0' \tau_{-n} \right)$, where $\widetilde \gamma_0'$ is the flipped transformation of $\gamma_0'$, as defined in (\ref{transformation}). Note that $\Tr( \gamma_n) \le \alpha + \beta$ and $\Tr( \gamma_n^\sharp) \le \alpha + \beta$, so that $I_{\alpha + \beta} \le \cE(\gamma_n)$ and $I_{\alpha + \beta} \le \cE(\widetilde \gamma)$ according to the fourth assertion of Lemma \ref{lemma5}. A straightforward calculation leads to
\begin{align*}
	\cE(\gamma_n) + \cE(\gamma_n^\sharp) & = 2 \cE(\gamma_0) + 2 \cE^\infty(\widetilde \gamma_0) - \dfrac{\beta(Z - \alpha)}{n} + O(\re^{-\delta n})  \\
		& = 2 I_\alpha + 2 I^\infty_{\beta} - \dfrac{\beta(Z - \alpha)}{n} + O(\re^{-\delta n}).
\end{align*}
For $n$ large enough, $-\beta(Z - \alpha) n^{-1} + O(\re^{\delta n})$ becomes negative. As before, either $\cE(\gamma_n)$ or $\cE(\gamma_n^\sharp)$ is strictly less than $I_\alpha + I_\beta^\infty$. Therefore, $I_{\alpha + \beta} < I_{\alpha} + I_{\beta}$. \\

Let us now assume that $\epsilon_F' = 0$. Then, there exists $\Psi \in H^2(\R^3, \C^2)$ such that $\| \Psi \|_{L^2} = 1$, $H^\infty_{\gamma_0'} \Psi = 0$ and $\gamma_0' \Psi = \mu \Psi$ with $\mu > 0$. Then, for $0 < \eta < \mu$, we introduce $\gamma_{\eta} = \gamma_0 + \eta | \Phi_{N_2+1} \ket \ \bra \Phi_{N_2+1} |$ and $\gamma'_{\eta} = \gamma_0' - \eta | \Psi \ket \ \bra \Psi |$, so that $\gamma_{\eta} \in \cP_{\alpha + \eta}$ and $\gamma'_\eta \in \cP_{\beta - \eta}$. Moreover,
\[
	\cE(\gamma_\eta) = \cE(\gamma_0) + 2 \eta \epsilon_{N_2+1} + o(\eta) = I_\alpha + 2 \eta \epsilon_{N_2+1} + o(\eta)
\]
and
\[
	\cE^\infty( \gamma'_\eta) = \cE^\infty(\gamma_0') + o(\eta) = I^\infty_\beta + o(\eta).
\]
Using the facts that $\gamma_0 + \eta | \Phi_{N_2+1} \ket \ \bra \Phi_{N_2+1} | \in \cP_{\alpha + \eta}$ and $\gamma_0' - \eta | \Psi \ket | \bra \Psi | \in \cP_{\beta - \eta}$, it holds that
\[
	I_{\alpha + \beta} \le I_{\alpha + \eta} + I^\infty_{\beta - \eta} \le \cE(\gamma_\eta) + \cE^\infty(\gamma'_\eta) \le I_\alpha + I^\infty_\beta + 2 \eta \epsilon_{N_2+1} + o(\eta).
\]
Because $\epsilon_{N_2+1} < 0$, for $\eta$ small enough, the left hand side is strictly less that $I_\alpha + I_\beta^\infty$, which concludes the proof.

\end{proof}


%


\section*{Acknowledgments}

I am grateful to E. Canc\`es for discussions and help throughout this work. This work was partially supported by the ANR MANIF.


\section*{References}

\bibliographystyle{unsrt.bst}
\bibliography{Art_ExistenceMin_LSDA_ArXiV}
\end{document}